\documentclass[conference]{IEEEtran}

\usepackage[utf8]{inputenc} 
\usepackage[T1]{fontenc}    
\usepackage{url}            
\usepackage{booktabs}       
\usepackage{amsfonts}       
\usepackage{nicefrac}       
\usepackage{microtype}      
\usepackage{xcolor}

\usepackage{algorithmic}
\usepackage{epsfig}
\usepackage{epstopdf}
\usepackage{graphicx}
\usepackage{latexsym}
\usepackage{amsmath}
\usepackage{amsfonts}
\usepackage{amssymb}

\usepackage{mathrsfs}
\usepackage{float}
\usepackage{pifont}
\usepackage{yhmath}

\usepackage{eufrak}
\usepackage{amsthm}
\usepackage{booktabs}
\usepackage{wasysym}
\usepackage{array}
\usepackage[ruled]{algorithm2e}

\usepackage{bigstrut,multirow,rotating}

\usepackage{graphicx}
\usepackage{caption}
\usepackage{subcaption}

\newtheorem{theorem}{Theorem}
\newtheorem{lemma}{Lemma}
\newtheorem{corollary}{Corollary}
\newtheorem{definition}{Definition}

\theoremstyle{remark}

\makeatletter
\newcommand\figcaption{\def\@captype{figure}\caption}
\newcommand\tabcaption{\def\@captype{table}\caption}
\makeatother

\usepackage{ulem}

\begin{document}
%
\title{
Distributed Decisions on Optimal Load Balancing\\ in Loss Networks 
}
%
%
%

\author{Qiong~Liu
        ,~Chenhao~Wang* and Ce Zheng 
\thanks{Q. Liu is with INFRE, Telecom Paris, France. email: Qiong.Liu@telecom-paris.fr}
\thanks{* C. Wang is the corresponding author, and is with Beijing Normal  University, China. email: chenhwang@bnu.edu.cn}
}

\author{
\IEEEauthorblockN{Qiong Liu\IEEEauthorrefmark{1}, Chenhao~Wang\IEEEauthorrefmark{2}, Ce Zheng\IEEEauthorrefmark{1}}
\IEEEauthorblockA{\IEEEauthorrefmark{1}Télécom Paris, Institut Polytechnique de Paris,  France}
\IEEEauthorblockA{\IEEEauthorrefmark{2}Beijing Normal  University, China}

\IEEEauthorblockA{Email: qiong.liu@elecom-paris.fr, chenhwang@bnu.edu.cn, ce.zheng@telecom-paris.fr }
}

\maketitle

\begin{abstract}
When multiple users share a common link in direct transmission, packet loss and link congestion may occur due to the simultaneous arrival of traffics at the source node. To tackle this problem, users may resort to an indirect path: the packet flows are first relayed through a sidelink to another source node, then transmitted to the destination. This behavior brings the problems of packet routing or load balancing: (1) how to maximize the total traffic in a collaborative way; (2) how self-interested users choose routing strategies to minimize their individual packet loss independently.
In this work, we propose a generalized mathematical framework to tackle the packet and load balancing issue in loss networks.  In centralized scenarios with a planner, we provide a polynomial-time  algorithm to compute the system optimum point where the total traffic rate is maximized. Conversely, in decentralized settings with autonomous users making distributed decisions, the system converges to an equilibrium where no user can reduce their loss probability through unilateral deviation. We thereby provide a full characterization of Nash equilibrium and examine the efficiency loss stemming from selfish behaviors, both theoretically and empirically. In general, the performance degradation caused by selfish behaviors is not catastrophic; however, this gap is not monotonic and can have extreme values in certain specific scenarios. 
\end{abstract}


\IEEEpeerreviewmaketitle

\section{Introduction}
Since the seminal work of Erlang \cite{erlang1917solution}, loss networks have played a crucial role in analyzing and optimizing stochastic systems involving simultaneous resource utilization, and non-backlogging workloads (for an extensive overview, see \cite{jung2019revisiting}). Meanwhile, in the post-5G era, cloud-enabled networks have emerged as a dominant architecture, where multiple servers collect data from users and relay it to a central hub for final processing. To guarantee network efficacy, that is no server is either overburdened or underutilized, load balancing strategies are well studied, e.g., \cite{cao2013optimal}. In this context, loss networks provide valuable mathematical frameworks for comprehending and enhancing load distribution within cloud-enabled networks.

Early load balancing research for cloud-enabled networks focused on centralized scenarios, where a centralized planner scheduled workloads to optimize aspects like performance-energy tradeoffs \cite{cao2013optimal} and algorithmic considerations \cite{ghomi2017load,zhao2015heuristic}. However, due to the stringent latency requirement for real-time decisions and the increasing signaling overhead caused by the large-scale deployment of servers and massive users, distributed decisions become a better solution. In this context, the complexity of the problem increases due to the  non-cooperative and competitive behaviors among users within the system.

To address the challenges of load balancing in a distributed way, game theory provides a mathematical framework that describes and analyzes scenarios with interactive decisions \cite{scutari2010convex}. Till now, some studies have demonstrated the efficacy of game-theoretic models in addressing load balancing problems. For instance, Mondal et al. \cite{mondal2020game} developed a game-theoretic model for load balancing among competitive cloudlets, while Yi et al. \cite{yi2020queueing} investigated a similar problem, incorporating additional considerations of queue-aware strategies.  In \cite{altman2014routing,toure2020congestion}, symmetric loss models where each source has an equal number of users are considered. However, previous studies mostly focused on limited cases of identical user strategies, which may not reflect real-world scenarios, i.e., different users may have different objectives and preferences. Therefore, further research is needed to develop game-theoretic models that can address the challenges of load balancing  in a more general and realistic manner.

In this paper, we employ game theory to address load balancing in both distributed and centralized environments, where users  have non-identical strategies and  the number of users is not evenly distributed. Specifically, we consider the load balancing in a cloud-enabled  network consisting of $m$ source nodes  (servers) $\{s_1,\ldots,s_m\}$ and one destination node (central hub) $d$. Each source $s_i$ has $n_i$ users seeking  service, and the traffic originating from each user is assumed to follow an independent Poisson point process with an identical rate. The  nodes in the network are connected by two types of communication links, namely sidelinks that connect two sources, and direct links that connect a source and destination. The sidelink has a random identical independent distribution (i.i.d) loss with a fixed probability $q$, and the direct link has a congestion loss that depends on the arrival rate and service rate of each server. 

The user cannot split its trafﬁc, and has to determine how to route all of its traffic from the source node arrived at to the destination node. There are two approaches for the traffic transmission: a direct path (DP) in which the packet goes directly from the source arrived at to the destination, and an indirect path (IP) in which the packet is first relayed to another source node and then takes the direct link from that node to the destination. 
We treat packet loss probability as the performance metric in load balancing, instead of additive costs like delay or fees in classical routing games \cite{toure2020congestion, patriksson2015traffic}, resulting in a non-additive and non-convex optimization process. Each user aims to minimize its own loss probability and engage in a game by strategically selecting its own path. In the end, no user can reduce its loss probability by unilateral deviation and reaches the state of \textit{Nash Equilibrium} (NE).

\subsection{Our Contributions}
Our work contributes to the load balancing game in the following aspects: First, we prove two lemmas related to the optimal solution when a centralized planner exists. Based on these lemmas, a low-complexity algorithm that maximizes the total traffic is proposed.
Second, we study the decentralized environment where decisions are made by autonomous and self-interested users. The sufficient and necessary conditions on NE are derived, which depend on the number of users on direct path and each indirect path. 
Moreover, since a NE may be suboptimal, we use the price of anarchy (PoA) \cite{koutsoupias1999worst} to measure the gap between the NE led by users' selfish behaviors and the system optimum achieved by the centralized planner. 

The rest of the paper is structured as follows. The formal model and notations are presented in Section \ref{sec:model}. In Section~\ref{sec:opt}, we provide details  to compute the optimal solution that maximizes the total traffic when a centralized planner exists. In Section~\ref{sec:ne_poa}, we study the NE in the decentralized decision-making scenarios, and analyzed the efficiency loss stemming from selfish behaviors. In Section \ref{sec:twopure}, a fine-grained analysis is performed on the existence of NE in various network configurations for a specific scenario involving two source nodes. 
Numerical results are presented and discussed in Section~\ref{sec:sim}.  Finally, Section~\ref{sec:con} concludes the paper and outlines some future work. 

\subsection{Other related works}
\emph{Routing games.} 
As a special class of congestion games, routing games in a network are problems of routing traffic to achieve
the best possible network performance, and  have been studied within various contexts
and within various communities, for example,  the mathematics community \cite{rosenthal1973class}, the telecommunications \cite{acemoglu2018informational}, and theoretical computer science \cite{chen2018equilibrium,chen2020efficiency}. The above references have all in common a cost framework which is additive over links,
such as delays or tolls, and is flow conserving (the amount entering a node equals the amount leaving it).  Routing games with non-additive cost in loss networks are studied in \cite{chowdhury2018non}.

\emph{Braess-like paradox in distributed systems.}
The  Braess-like paradox is said to occur in a network system with distributed behaviors if adding an extra link or adding communication capacity to the system leads to a worse system performance. It widely exists in transportation networks and queuing networks. Bean {et al.} \cite{bean1997braess} show that it can occur in loss networks. 
Kameda {et al.} \cite{kameda2000braess} consider a model similar to ours in that a job (packet) can be processed directly or indirectly; however, they do not consider the loss probability.  They identify a Braess-like paradox in which adding capacity to the channel may degrade the system performance on the response time.
Kameda and Pourtallier \cite{kameda2002paradoxes} characterize conditions under which such paradoxical behavior occurs, and give examples in which the degradation of performance may increase without bound.

\section{Model and Preliminaries}\label{sec:model}
We abstractly model our problem using a graph. Consider a network with $m$ source nodes $S=\{s_1,\ldots,s_m\}$ and one destination node $d$.  For each source node $s_i\in S$, let $N_i$ be the set of users arriving at $s_i$, and $n_i=|N_i|$ be the number of such users. Without loss of generality, we assume $n_1 \!> \!n_2 \!>\! \dots\!>\!n_m$. Denote $[m]=\{1,\ldots,m\}$. There is a total of $n=\sum_{i\in[m]}n_i$ users in the system, who are self-interested players in the game. We say players and users interchangeably throughout this paper. Each user is identified with a flow (or traffic) of packets, which originates from the user and is assumed to form an independent Poisson process with an identical rate $\phi$. See  Fig.~\ref{fig:ill} for illustration. 


%

 \begin{figure}
    \centering
    \includegraphics[scale=0.6]{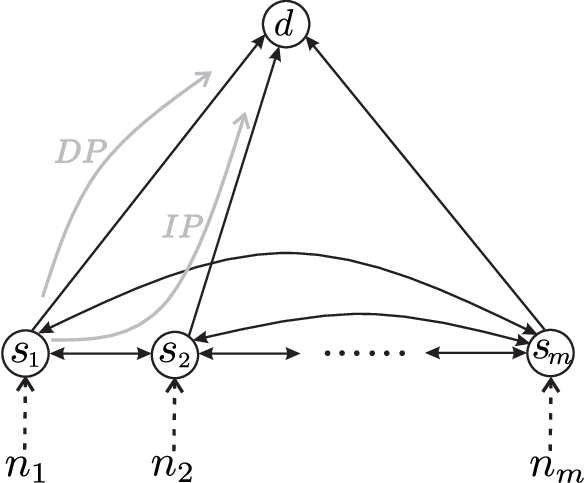}
    \caption{An illustration of the game model.}
    \label{fig:ill}
\end{figure}

Each user controls its route that all its packets should follow. For a user associated with $s_i\in S$, there are only two types of routes to ship these packets to the destination $d$:  either a direct path (DP)  $(s_i,d)$, or  an indirect two-hop path (IP) $(s_i,s_j,d)$ for some $s_j\neq s_i$, in which the packet is first sent to another source $s_j$ by the side link $(s_i,s_j)$, and then passes through the direct link $(s_j,d)$.  



\smallskip\textbf{Strategies.} For every source $s_i$, each user $k\in N_i$ decides a one-shot \textit{strategy} $\mathbf p_k^{(i)}=\left( p_{k1}^{(i)},\ldots,p_{km}^{(i)} \right)^T\in [0,1]^m$ with $\sum_{j\in[m]}p_{kj}^{(i)}=1$, where $p_{ki}^{(i)}$ is the probability of routing all packets through DP, and $p_{kj}^{(i)}~(j\in[m],j\neq i)$ is the probability of routing all packets through IP $(s_i,s_j,d)$. 
{When no confusion arises, we simply write the strategy $\mathbf p_k^{(i)}$ as $\mathbf p_k$.}
We focus on \textit{pure strategies} in this paper: a strategy $\mathbf p_k$ is \textit{pure} if {$\|\mathbf p_k\|_{\infty}=1$}, i.e., user $k$ deterministically selects a route with probability 1 (for example, $\mathbf p_k=(0,1,0,\ldots,0)^T$). 
Let $\mathbf p=(\mathbf p_1,\ldots,\mathbf p_n)$ be the strategy profile of all users. 







\smallskip\textbf{Loss probability and loss rate.}
There are two types of losses: (1) 
Losses on side  links. We assume that a packet originating from node $s_i$ and relayed to node $s_j$ is lost with a fixed probability $q$ for every side  link $(s_i,s_j)$, independently of any other loss. Denote by $\bar q=1-q$ the probability that a packet is successfully relayed.  (2) Congestion losses on direct links. We assume that there is no buffer to restore the backlogged packets, {so a packet will be lost when it enters the direct path which is occupied for the transmission of another packet.} The transmission time of a packet {on} a direct link $(s_i,d)$ is a random variable $\sigma$ following a distribution $\mathcal X$, {which is assumed to be an identically independent distribution (i.i.d) for all packets}.  


Given  strategy profile $\mathbf p$,  user $k\in N_i$ continuously sends packets that follow an independent Poisson process with rate $p_{ki}^{(i)}\cdot\phi$ to DP $(s_i,d)$, and an independent Poisson process of packets with rate $p_{kj}^{(i)}\cdot\phi$ to IP $(s_i,s_j,d)$, for any $s_j\neq s_i$. 
{Since there is a random
loss on} the side  link $(s_i,s_j)$, the flow of packets from user $k\in N_i$ that arrive at the node $s_j$ is also a Poisson process with rate  $\bar{q}p_{kj}^{(i)}\phi$. 




Thus, for each source $s_i\in S$, the flow over the link $(s_i,d)$ is Poisson distributed with a
traffic rate $T_i(\mathbf{p})$ given by
\begin{align}\label{eq:ti}
    T_i(\mathbf p)=\sum_{k\in N_i}p_{ki}^{(i)} \phi + \sum_{j\in[m]\backslash\{i\}}\sum_{k\in N_j} p_{ki}^{(j)} \bar{q}\phi.
\end{align}

\noindent When no confusion arises, we simply write $T_i(\mathbf p)$ as $T_i$. 

The probability of no {congestion} loss on the direct link $(s_i,d)$  equals the probability that there is no arrival during a transmission time $\sigma$, which is given by
$\operatorname*{\mathbb E}_{\sigma\sim \mathcal X}e^{-T_i\sigma}$.
As usual  practice, assume $\mathcal X$ is an exponential distribution with a rate parameter $\mu$ (service rate) and mean $\frac{1}{\mu}$. Thus the probability of no {congestion} loss on $(s_i,d)$ is
\begin{align}\label{eq:loss_prob_direct}
    \operatorname*{\mathbb E}_{\sigma\sim \mathcal X}e^{-T_i\sigma}=\int_0^{+\infty}\mu e^{-\mu\sigma} e^{-T_i\sigma}d\sigma=\frac{\mu}{T_i+\mu},
\end{align}

\noindent and the loss probability  on link $(s_i,d)$ is $\frac{T_i}{T_i+\mu}$.



{
Given the strategy profile $\mathbf p$, for $s_i\in S$ and $k\in N_i$, the \textit{loss rate} of user $k$ is defined as
\begin{align}
\label{eq:lossrate}
    LR_k(\! \mathbf{p} \!) \!\!=\!\!\! \left[\! p_{ki}^{(i)}\frac{T_i}{T_i \!+\! \mu} \!\!+\!\! \left(\! 1 \!-\! p_{ki}^{(i)} \! \right) \! q \!+\!
    (\! 1 \!-\! q)  \!\!\!\!\!\!\!\!  \sum_{j\in[m]\backslash\{i\}} \!\!\!\!\!\!\!  p_{kj}^{(i)}\frac{T_j}{T_j \!+\! \mu} \!\right] \!\! \phi,
\end{align}
and the loss probability of user $k$ is $\frac{LR_k(\mathbf p)}{\phi}$.

}

\smallskip\textbf{Total traffic.} Regarding the system efficiency, we measure it by the \textit{total traffic rate}  arriving at the destination $d$. Given the strategy profile $\mathbf p$, 
 the total traffic rate $TR(\mathbf p)$ of the system can be derived in two ways. The first expression is derived as the summation of successful transmission rates on direct links:
\begin{align}
\label{equ:29}
    &TR(\mathbf p)=\sum_{i\in [m]}T_i\cdot\frac{\mu}{T_i+\mu} \notag \\
       &=\mu \!\left[\! m-\sum_{i\in [m]}\frac{\mu}{\sum\limits_{k\in N_i} p_{ki}^{(i)}\phi +\!\!\! \sum\limits_{ j\in[m] \backslash \{i\} }\sum\limits_{k\in N_j}p_{ki}^{(j)}\bar{q}\phi+\mu} \right]
\end{align}
where $T_i$ is the traffic rate over link $(s_i,d)$, and $\frac{\mu}{T_i+\mu}$ is the probability of no congestion loss on $(s_i,d)$.




The second expression is from users' perspective:
\begin{align}\label{equ:hk}
    TR(\mathbf p)&:=\sum_{i\in [m]}\sum_{k\in N_i}(\phi-LR_k),
\end{align}
where $\phi-LR_k(\mathbf p)$ is the traffic rate of user $k\in N_i$ that successfully arrive at $d$. It is not hard to see that (\ref{equ:29}) and (\ref{equ:hk}) are equivalent.



\smallskip\textbf{Nash equilibria.} A Nash equilibrium (NE) is a strategy profile where no player can decrease its loss probability by unilaterally deviating to any other strategy. Formally, we give a definition.
\begin{definition}\label{def:ne}
A strategy profile $\mathbf p$ is a \textit{Nash equilibrium}, if for any source $s_i\in S$ and any player $k\in N_i$, we have
$$LR_k(\mathbf p_k,\mathbf p_{-k})\le LR_k(\mathbf p_k',\mathbf p_{-k}),$$
where $\mathbf p_k'$ can be any feasible strategy of player $k$, and  $\mathbf p_{-k}$ is the strategy profile of all other players.
\end{definition}




We measure the efficiency of NEs by the price of anarchy (PoA) \cite{koutsoupias1999worst},  which is defined as the ratio between social efficiencies in an optimal solution and in the worst NE. 
Formally, given an instance $\Gamma$ of this game, {we} define
$$PoA(\Gamma)=\frac{TR(opt)}{\min_{\mathbf p\in\mathbb{NE}}TR(\mathbf p)}.$$
where $opt$ is an optimal solution of $\Gamma$, and $\mathbb{NE}$ is the set of all NEs. The PoA of the whole game is defined as the maximum over all instances, that is, $PoA=\max_{\Gamma}PoA(\Gamma)$.

\section{Centralized Analysis}
\label{sec:opt}

%

The main technical results of the paper are presented now. We show how to compute an optimal solution that maximizes the total traffic.
Note that the total traffic rate depends on the number of users working on each source by DP or IP, but not the users' identity.  
Given a strategy profile $\mathbf p$, let $u_i=|\{k\in N_i~|~p_{ki}^{(i)}=1\}|$ be the number of users working with DP $(s_i,d)$, and let $v_i=|\{k\in N_j, j \in [m] \backslash {i}|~p_{ki}^{(j)}=1\}|$ be the users working with IP through link $(s_i,d)$. Define $y_i=u_i+v_i$ as the number of users who choose source $s_i$ (including both DP and IP).



\begin{lemma}\label{lem:both}
In any optimal solution, for any source $s_i$,  either $u_i=n_i$ or $v_i=0$ or both hold. 
\end{lemma}



\begin{proof}
 Let $\mathbf p$ be an optimal solution. Suppose for contradiction that $u_i<n_i,v_i>0$ for some source $s_i$. Then there exists a user (say, $k$) in $N_i$ who chooses IP (say, $(s_i,s_{i'},d)$ for some $i'\neq i$). Also, since $v_i>0$, there exist a source $s_j\neq s_i$ and a user $l\in N_j$ who chooses IP $(s_j,s_i,d)$. The total traffic rate is $TR(\mathbf p)=\frac{\mu T_i}{T_i+\mu}+\sum_{w\in [m]\backslash\{i\}}\left(\frac{\mu T_w}{T_w+\mu}\right)$.

Now we show that  the total traffic rate can be improved by revising $\mathbf p$. Let user $k\in N_i$ choose DP, and let user $l\in N_j$ choose IP $(s_j,s_{i'},d)$. Fixing all others' strategies, denote the new strategy profile by $\mathbf p'$, and define $u_i',v_i'$ accordingly. Note that $u_i'=u_i+1,v_i'=v_i-1$, and $T_w(\mathbf p')=T_w(\mathbf p)$ for all source $s_w\neq s_i$. Since $q>0$, we have 
$$T_i(\mathbf p')=(u_i+1)\phi+(v_i-1)\bar{q}\phi> u_i\phi+v_i\bar{q}\phi=T_i(\mathbf p). $$
So $TR(\mathbf p')>TR(\mathbf p)$, contradicting to the optimality.
\end{proof}

Lemma~\ref{lem:both} indicates that if a source (say $s_i$) provides service to the  users of other sources, then all users of $s_i$ choose DP.

\begin{lemma}\label{lem:opt}
In any optimal solution, there must exist $\tilde i\in[m]$, such that $v_l=0$ for all $l\le \tilde i$, and $u_j=n_j$ for all $j>\tilde i$.  
\end{lemma}
\begin{proof}
Given an optimal solution $\mathbf p$, suppose for contradiction that there exist $i,j\in[m]$ ($i<j$) such that $v_i>0$ and $u_j<n_j$. By Lemma \ref{lem:both}, we have $u_i=n_i$ and $v_j=0$.  There exists a source $s_{i'}$ and a user $k\in N_{i'}$  selecting IP $(s_{i'},s_i,d)$. There exists a source $s_{j'}$ ($j'\neq j$) and a user $k'\in N_j$ selecting IP $(s_j,s_{j'},d)$. Note that when $i'=j$ and $j'=i$, users $k$ and $k'$ may coincide. The total traffic rate is 
\begin{align*}
    TR(\mathbf p)&=\frac{\mu T_i}{T_i+\mu}+\frac{\mu T_j}{T_j+\mu}+\sum_{w\in [m]\backslash\{i,j\}}\left(\frac{\mu T_w}{T_w+\mu}\right)\\
    &=\mu(2-\frac{1}{T_i+\mu}-\frac{1}{T_j+\mu})+\sum_{w\in [m]\backslash\{i,j\}}\left(\frac{\mu T_w}{T_w+\mu}\right).
\end{align*}

Now we show that the total traffic rate can be improved by revising $\mathbf p$. Let user $k$ choose IP $(s_{i'},s_{j'},d)$ if $i'\neq j'$ and choose DP $(s_{i'},d)$ if $i'=j'$. Let user $k'\in N_j$ choose DP. Fixing all others’ strategies, denote the new strategy profile by $\mathbf p'$, and define $u_i',v_i'$ accordingly. Note that $v_i' \!=\! v_i-1$, $u_j' \!=\! u_j+1$, and $T_w(\mathbf p')=T_w(\mathbf p)$ for all other sources $s_w \!\neq\! s_i,s_j$. Since $i \!<\! j$, it follows that $n_i \!\ge\! n_j \!>\! u_j$. Therefore, we have
\begin{align*}
    &\frac{1}{T_i+\mu}+\frac{1}{T_j+\mu}=\frac{1}{n_i\phi+v_i\bar{q}\phi+\mu}+\frac{1}{u_j\phi+\mu}\\
    >~&\frac{1}{n_i\phi+v_i'\bar{q}\phi+\mu}+\frac{1}{u_j'\phi+\mu}
    =\frac{1}{T_i'+\mu}+\frac{1}{T_j'+\mu},
\end{align*}
which indicates that $TR(\mathbf p)<TR(\mathbf p')$, a contradiction.
\end{proof}

Lemma~\ref{lem:opt} shows that there exists a threshold $\tilde i$:
1) if $i>\tilde i$, all users from $s_i$ chose DP;
2) if $i\leq\tilde i$,  portion users chose DP, and portion users chose IP. 

Now we are ready to present Algorithm~\ref{alg:heuristic}. The main idea is searching for $\tilde i$ in Lemma \ref{lem:opt}. For each candidate of $\tilde i$, let $B$ be the number of users selecting IP, all of whom come from $L\!=\!\{s_l~|~l\le \tilde i\}$, and go to $R\!=\!\{s_j~|~j>\tilde i\}$. For every possible value of $B$, we compute the best possible way for extracting the $B$ users from $L$ and distributing them over $R$. 








\begin{algorithm}
	\caption{\hspace{-4pt}{ \bf Computing an optimal solution.}}
	\label{alg:heuristic}
	\begin{algorithmic}[1] 
	\REQUIRE $m$ source with $n_1\ge n_2\ge\cdots\ge n_m$, $\phi,\mu,q$ 
	\ENSURE $(u^*_i,v^*_i)_{i\in N}$
	\STATE Initialize $u_i=v_i=0$ for all $i\in[m]$. $TR^*=0$
	\FOR{$\tilde i=1,\ldots,m$}
	\STATE Let $v_l=0$ for all $l=1,2,\ldots,\tilde i$
	\STATE Let $u_j=n_j$ for all $j=\tilde i+1,\ldots,m$
	\FOR{$B=1,2,\ldots,\sum_{l=1}^{\tilde i}n_l$}
	\STATE (a) Compute $(u_l)_{l\in[\tilde i]}$ s.t. $\sum_{l\!=\!1}^{\tilde i}(n_l\!-\!u_l)\!=\!B$ \\and the values of $u_l$ are as equal as possible.
	\STATE (b) Compute $(v_j)_{j>\tilde i}$ s.t. $\sum_{j=\tilde i+1}^mv_j\!=\!B$ and the values of $n_j\phi\!+\!v_j\bar{q}\phi\!+\!\mu$ are as equal as possible.
	\STATE (c) Compute $TR$ with respect to $(u_i,v_i)_{i\in N}$
	\IF{$TR>TR^*$}
	    \STATE $TR^*\leftarrow TR$, and $(u^*_i,v^*_i)_{i\in N}\leftarrow (u_i,v_i)_{i\in N}$
	\ENDIF
	\ENDFOR
	\ENDFOR
	\end{algorithmic}
\end{algorithm}

In Algorithm~\ref{alg:heuristic}, step (a) is to make $T_l$ (and thus no congestion probability $\frac{\mu}{T_l+\mu}$) as equal as possible for $l\!\in\![\tilde i]$. This can be realized by initializing $u_l\!=\!n_l$, and then removing players one by one from the highest $u_l$ and updating until $B$ players have been removed. The goal of step (b) is to make $T_j$ (and thus $\frac{\mu}{T_j+\mu}$) as even as possible. This can be realized by initializing $v_j=0$,  then adding users one by one to $v_{j'}=\arg\min_{j>\tilde i}n_j\phi+v_j\bar{q}\phi+\mu$ and updating, until $B$ players have been added. These two steps guarantee that the $B$ loads are distributed in an optimal way to maximize the traffic rate.

\begin{theorem}\label{thm:main}
Algorithm \ref{alg:heuristic} returns an optimal solution for maximizing the total traffic, and runs in $O(mn^2)$ time.
\end{theorem}
\begin{proof}
In the first loop, we traverse all indexes in $[m]$ to find the $\tilde i$ in Lemma \ref{lem:opt}. In the second loop, we traverse all possible numbers of users who select IP, and given any such a number  $B$, we extract the $B$ users from $\{s_l~|~l\le \tilde i\}$  and distribute them over $\{s_j~|~j>\tilde i\}$ in an optimal way to maximize the traffic rate. So all possible optimal solutions have been searched by the algorithm, giving the optimality. 
For the time complexity, we have $m$ iterations in the first loop, at most $n$ iterations in the second loop, and the time for each iteration is $O(n)$. 
\end{proof}

Intuitively, when the transmission loss probability is sufficiently large, all packets should go through DP; when there is no transmission loss, the load of packets should be distributed evenly over all sources. We verify the intuition as follows.
\begin{corollary}
If $q=1$, the unique optimal solution is that all users choose DP (i.e., $u_i=n_i, v_i=0, \forall i\in M$). If $q=0$, a strategy profile $\mathbf p$ is optimal if and only if $|y_i-y_j|\le 1$ for all $i,j\in [m]$. 
\end{corollary}
\begin{proof}
If $q=1$,  $TR=\sum_{i\in[m]}\frac{\mu u_i\phi}{u_i\phi+\mu}$ is increasing with respect to every $u_i$. By the monotonicity,  the optimum is achieved when $u_i=n_i$. If $q=0$, suppose for contradiction that there exist $i,j\in [m]$ in an optimal solution $\mathbf p$ such that $y_i-y_j\ge 2$. The total traffic rate is $TR=\mu(m-\sum_{k\in [m]}\frac{\mu}{y_k\phi+\mu})$. Consider a new strategy profile $\mathbf p'$ with $y_i'=y_i-1,y_j'=y_j+1$, i.e., a user who chooses source $s_i$ deviates to $s_j$.  Then the total traffic rate becomes $TR'=\mu(m-\sum_{k\in[m]\backslash\{i,j\}}\frac{\mu}{y_k\phi+\mu}-\frac{\mu}{y_i'\phi+\mu}-\frac{\mu}{y_j'\phi+\mu})>TR$, a contradiction.
\end{proof}

\section{Decentralized Analysis}\label{sec:ne_poa}
In this section, we study the Nash equilibria in the decentralized decision-making scenario where each user makes a decision on the choice of DP or IP. 

 \vspace{-2mm}

\subsection{Characterization of NEs}\label{subsec:ne}

A NE should satisfy that: for a user selecting DP, its loss rate will not decrease if it deviates to any IP; for a user selecting IP, its loss rate will not decrease if it deviates to DP or another IP.
We formalize it as the following characterization.

\begin{theorem}\label{theo:msps}
Given an arbitrary strategy profile $\mathbf p$ with $(u_i,v_i)_{i\in [m]}$, let $i^*\in\arg\min_{i\in[m]}\{u_i+v_i\bar{q}\}$, and let $x_{ij}\in\{0,1\}$ be an indicator where $x_{ij}=1$ if there exists at least one user selecting IP $(s_i,s_j,d)$. Then, $\mathbf p$ is a NE, if and only if the following conditions are satisfied: 
\begin{itemize}
    \item[(i)] for all $i\in[m]$ with $u_i>0$, 
    we have
     \vspace{-2mm}
    \begin{equation}\label{con1}
   \bar{q}(u_i+v_i\bar{q})\le u_{i^*}+v_{i^*}\bar{q}+\bar{q}+\frac{q\mu}{\phi};
    \end{equation}
    \item[(ii)] for all $i,l\in[m]$ with $x_{il}=1$, we have
      \begin{align}\label{con2}
          u_l \!+\! v_l\bar{q} \leq  \min \!\left\{\! \bar{q}(\! u_i \!+\! 1 \!+\! v_i\bar{q}) \!-\! \frac{q\mu}{\phi}, u_{i^*} \!+\! v_{i^*}\bar{q} \!+\! \bar{q} \!\right\}
      \end{align}
\end{itemize}
\end{theorem}

\begin{proof}
    Suppose $\mathbf p$ is a NE. Consider an arbitrary source $s_i\in S$ and arbitrary user $k\in N_i$.

    \textbf{Case 1.} User $k$ selects DP in $\mathbf p$ (denoted as $\mathbf{p}_k^{(i)}$ where $p_{ki}^{(i)}=1$).  If it deviates to IP $(s_i,s_j,d)$ where $j \neq i$ (denoted as $\mathbf{p'}_k^{(i)}$ where $p_{kj}^{i}=1$), by Definition \ref{def:ne}, we have $LR_k(\mathbf p_k^{(i)},\mathbf p_{-k})\le LR_k(\mathbf {p'}_k^{(i)},\mathbf p_{-k})$. It is equivalent to 
    \begin{align*}
    &   1\!-\!\frac{\mu}{u_i\phi+v_i\bar{q}\phi+\mu}\!\le \! q+\bar{q}\cdot\left(1\!-\!\frac{\mu}{u_j\phi+(v_j+1)\bar{q}\phi+\mu}\right)\\
  \Leftrightarrow &~ \frac{\bar{q}}{u_j\phi+(v_j+1)\bar{q}\phi+\mu}\le \frac{1}{u_i\phi+v_i\bar{q}\phi+\mu}\\
    \Leftrightarrow &~ \bar{q}(u_i+v_i\bar{q})-\frac{q\mu}{\phi}\le u_j+(v_j+1)\bar{q},
\end{align*}
 The above inequality should hold for all $j\neq i$, and thus is equivalent to  Equation (\ref{con1}). 

    \textbf{Case 2.} User $k$ selects IP $(s_i,s_l,d)$ in $\mathbf p$. If it deviates to DP $(s_i,d)$, by Definition \ref{def:ne}, we should have $LR_k(\mathbf p_k^{(i)},\mathbf p_{-k})\le LR_k(\mathbf{p'}_k^{(i)},\mathbf p_{-k})$. It is equivalent to 
\begin{align*}
  &~ q\!+\!\bar{q}\cdot\left(1\!-\!\frac{\mu}{u_l\phi+v_l\bar{q}\phi+\mu}\right)\le 1\!-\! \frac{\mu}{(u_i+1)\phi+v_i\bar{q}\phi+\mu}\\
  \Leftrightarrow &~ \frac{1}{(u_i+1)\phi+v_i\bar{q}\phi+\mu}\le \frac{\bar{q}}{u_l\phi+v_l\bar{q}\phi+\mu}\\
    \Leftrightarrow &~ u_l+v_l\bar{q}\le \bar{q}(u_i+1+v_i\bar{q})-\frac{q\mu}{\phi}. 
\end{align*}
 Moreover, a NE must guarantee that user $k$ will  not deviate to another IP $(s_i,s_j,0)$, and thus we should have
    \begin{align*}
&~ 1-\frac{\mu}{u_l\phi+v_l\bar{q}\phi+\mu}\le 1-\frac{\mu}{u_j\phi+(v_j+1)\bar{q}\phi+\mu}\\
  \Leftrightarrow & ~u_l+v_l\bar{q}\le u_j+(v_j+1)\bar{q}.
\end{align*}
 Note that the above inequality should hold for all $j\neq i,l$. Therefore, we obtain Equation  \eqref{con2}.
\end{proof}

\subsection{Price of Anarchy}\label{subsec:poa}
We investigate the price of anarchy in this section, which measures the efficiency of NE. We  give an upper bound on the optimal total traffic rate, and  a lower bound on the total traffic rate of any NE.

\begin{lemma}\label{lem:kdd}
In an optimal solution $\mathbf p$, 
the total traffic rate is $TR_{\text{}}(\mathbf p)\le \mu m \left(1-\frac{\mu}{n_1\phi+\mu}\right)$.
\end{lemma}
\begin{proof}
Let $i$ be the index stated in Lemma \ref{lem:opt}. It suffices to show  with the proof by contradiction that in the optimal solution $\mathbf{p}$, $u_i+v_i\bar{q}\le n_i$. First, for $i=1$, we have $v_1\!=\!0$, and thus it satisfies $u_1+v_1\bar{q}\!=\!u_1\le n_1$. For any $i>1$, suppose for contradiction that $u_i+v_i\bar{q}>n_i$. Then $v_i>0$, and there exists a source $s_j$ and a user $k\in N_j$ that chooses the IP $(s_j,s_i,d)$, i.e., $u_j<n_j$. By Lemma \ref{lem:both}, it must be $v_j=0$, and thus $T_j=u_j\phi$. Denote the strategy as $\mathbf{p}$ with $p_{ki}^{(j)}=1$.
The total traffic rate is
\begin{align}
    TR(\mathbf p)=\mu \left( m - \frac{\mu}{T_j + \mu} - \frac{\mu}{T_i + \mu} -  \sum_{w\in[m]\backslash\{j\}} \frac{\mu}{T_w + \mu} \right)
\end{align}
We show that the total traffic rate can be improved with user $k \in N_j$ deviating from IP $(s_j, s_i, d)$ to DP $(s_j,d)$.
Fixing the strategies of all others, denote by $\mathbf p'$ the new strategy profile, and define $(u'_w,v'_w,T'_w)_{w\in[m]}$ accordingly. Note that $u_j'=u_j+1$, $v_i'=v_i-1$, $u_i'=u_i$, and $T_w'=T_w$ for any $w\in[m]\backslash\{j\}$. Since $u_i+v_i\bar{q}>n_1\ge n_j\ge u_j$, we have  
\begin{align*}
    &\frac{1}{T_j+\mu}+\frac{1}{T_i+\mu}=\frac{1}{u_j\phi+\mu}+\frac{1}{u_i\phi+v_i\bar{q}\phi+\mu}\\
    > &~\frac{1}{u_j'\phi+\mu}+\frac{1}{u_i'\phi+v_i'\bar{q}\phi+\mu}= \frac{1}{T_j'+\mu}+\frac{1}{T_i'+\mu}. 
\end{align*}
It indicates that $TR(\mathbf p')\!>\!TR(\mathbf p)$ is a contradiction. Consequently,  $u_i+v_i\bar{q}\!\leq\! n_i \!\leq \!n_1$. According \eqref{equ:29}, we have  $TR_{\text{}}(\mathbf p)\le \mu m(1-\frac{\mu}{n_1\phi+\mu})$.
\end{proof}

\begin{lemma}\label{lem:kgg}
Let $z=\min\{n_m,\frac{n}{4m}-\bar{q}-\frac{q\mu}{\phi}\}$. For every NE $\mathbf p$, the total traffic rate satisfies $TR(\mathbf p)\ge \mu(m-\frac{m\mu}{z\phi+\mu})$.
\end{lemma}
\begin{proof}
Let $i^*=\arg\min_{i\in[m]}\{u_i+v_i\bar{q}\}$. Since $TR(\mathbf p)\ge \mu m(1-\frac{\mu}{(u_{i^*}+v_{i^*}\bar{q})\phi+\mu})$, it suffices to prove that $u_{i^*}+v_{i^*}\bar{q}\ge z$. If $u_{i^*}+v_{i^*}\bar{q}\ge n_m$, it is done.  We only need to consider the case when $u_{i^*}+v_{i^*}\bar{q}< n_m\le n_{i^*}$. There exists some users in $N_{i^*}$ selecting IP. By Equation (\ref{con2}), we have 
$$u_{i^*}+v_{i^*}\bar{q}\le \bar{q}(u_{i^*}+1+v_{i^*}\bar{q})-\frac{q\mu}{\phi}.$$
By Theorem \ref{theo:msps}, for each $i\in[m]$, if $u_i>0$, then $\bar{q}(u_i+v_i\bar{q})\le u_{i^*}+v_{i^*}\bar{q}+\bar{q}+\frac{q\mu}{\phi}$; if $v_i>0$, then $u_i+v_i\bar{q}\le u_{i^*}+v_{i^*}\bar{q}+\bar{q}$. {In both cases, we obtain $u_i+v_i/2\le 2(u_{i^*}+v_{i^*}\bar{q}+\bar{q}+\frac{q\mu}{\phi})$.} 
Summing up over all $i\in[m]$, we have 
$$n/2\le \sum_{i\in[m]}(u_i+v_i/2)\le 2m (u_{i^*}+v_{i^*}\bar{q}+\bar{q}+\frac{q\mu}{\phi}),$$
which implies that $u_{i^*}+v_{i^*}\bar{q}\ge \frac{n}{4m}-\bar{q}-\frac{q\mu}{\phi}$. 
\end{proof}

\begin{theorem}
For any instance with $m$ sources, the price of anarchy is $PoA\le 1+\frac{n_1\mu}{n_1z\phi+z\mu}$, where $z=min\{n_m,\frac{n}{4m}-\bar{q}-\frac{q\mu}{\phi}\}$.
\end{theorem}
\begin{proof}
Combining the upper bound  in Lemma \ref{lem:kdd} and the lower bound on $TR(\mathbf p)$ for any NE $\mathbf p$ in Lemma \ref{lem:kgg}, it follows
$$PoA \le \frac{m - \frac{m\mu}{n_1\phi + \mu}}{m - \frac{m\mu}{z\phi + \mu}} = \frac{n_1(z\phi+\mu)}{z(n_1\phi+\mu)} = 1 + \frac{n_1\mu}{n_1z\phi + z\mu}.$$
\end{proof}


\section{A Particular Case: Two Sources}
\label{sec:twopure}

In this section, we focus on the special case of $m=2$. That is, there are only two sources $s_1$ and $s_2$. Assume w.l.o.g. that $n_1\ge n_2$. For each user $k \in N_i$, there is only one IP. Accordingly, its strategy becomes $\mathbf p_k^{(i)} = \left(p_{k1}^{(i)},p_{k2}^{(i)}\right), i = 1, 2$. 
And we have 
\begin{align}\label{eq:2sources}
    n_1 = u_1 + v_2; \\
    n_2 = u_2 + v_1.
\end{align}
The traffic rate $T_i(\mathbf p)$ in \eqref{eq:ti} is rephrased as
\begin{align}\label{eq:ti2}
    T_i(\mathbf p) = \sum_{k\in N_i}p_{ki}^{(i)} \phi + \sum_{k\in N_j, j\neq i} p_{ki}^{(j)} \bar{q}\phi = u_i \phi + v_i \bar{q} \phi.
\end{align}





Given strategy profile $\mathbf p$, the set $N$ is further partitioned into 4 subsets $(V_1,V_2,V_3,V_4)$ where {$V_1=\{k\in N_1~|~\mathbf p_k^{(1)} =(1, 0)\}$,  $V_2=\{k\in N_1~|~\mathbf p_k^{(1)} = (0,1)\}$, $V_3=\{k\in N_2~|~\mathbf p_k^{(2)}=(0,1)\}$ and $V_4=\{k\in N_2~|~\mathbf p_k^{(2)}=(1,0)\}$}.
Clearly, users in $V_1$ and $V_3$ choose DP, and users in $V_2$ and $V_4$ choose IP. 

Suppose $\mathbf p$ is a NE. We study the deviation of users in $V_1,V_2,V_3,V_4$, respectively. 
For user $k\in V_1$, the strategy is $\mathbf p_{k}^{(1)} = \left(p_{k1}^{(1)}, p_{k2}^{(1)} \right)=(1, 0)$, and the loss rate in \eqref{eq:lossrate} is
\begin{align}
    LR_k( \mathbf{p} ) 
    &\!=\! \frac{\phi T_1(\mathbf p)}{T_1(\mathbf p)+\mu} = \left[1 - \frac{\mu/\phi}{u_1 + v_1\bar{q} + \mu/\phi} \right] \phi.  \nonumber
\end{align}
When user {$k\in V_1$} deviates to IP, the strategy profile becomes {$\mathbf{p'} = \left( \mathbf{p'}_{k}^{(1)}, \mathbf{p}_{-k}\right)$ where $\mathbf p_{k}'^{(1)} =(0, 1)$.} 
The loss rate of user $k$ becomes
\begin{align*}
    LR_i(\mathbf p') \!=\! q\phi \!+\! \bar{q}\phi\frac{T_2(\mathbf p')}{T_2(\mathbf p') \!+\! \mu}
    \!=\! \left[1 \!-\! \frac{\bar{q}\mu/\phi}{u_2 \!+\! (v_2 \!+\! 1)\bar{q} \!+\! \mu/\phi}\right] \phi.
\end{align*}
Since $\mathbf p$ is NE, {$k$} has no incentive to deviate, and thus $LP_k(\mathbf p)\le LR_k(\mathbf p')$, 
which is equivalent to 
\begin{align}\label{eq:c1_}
t_1(u_2) : = \frac{q\mu/\phi \!+\! u_2(1 \!+\! \bar{q}^2) \!+\! (n_1 \!+\! 1)\bar{q}\!-\!n_2\bar{q}^2}{2\bar{q}} \geq u_1,
\end{align}
where $t_1(u_2)$ is a function with respect to variable $u_2$. 

For {user $k\in V_2$ with strategy $\mathbf p_{k}^{(1)} =(0, 1)$, the loss rate is
\begin{align*}
    LR_k(\mathbf p) = q\phi + \bar{q}\phi\frac{T_2(\mathbf p)}{T_2(\mathbf p) + \mu} = \left[1 - \frac{\bar{q}\mu/\phi}{u_2 + v_2\bar{q} + \mu/\phi}\right] \phi.
\end{align*}
}
When user {$k\in V_2$} deviates to DP, the strategy profile becomes {$\mathbf{p'} = \left( \mathbf{p'}_{k}^{(1)}, \mathbf{p}_{-k} \right)$ where $\mathbf p_{k}'^{(1)}=(1, 0)$. 
The loss rate of $k$ becomes
\begin{align*}
    LR_k(\mathbf p') = \frac{\phi T_1(\mathbf p')}{T_1(\mathbf p') + \mu} = \left[ 1 - \frac{\mu/\phi}{u_1+1+v_1\bar{q}+\mu/\phi} \right]\phi.    
\end{align*}
Since $\mathbf p$ is NE, we have $LR_k(\mathbf p)\le LR_k(\mathbf p')$, that is,  
\begin{equation}
    u_1\ge\frac{q\mu/\phi \!+\! u_2(1 \!+\! \bar{q}^2) \!+\! (n_1 \!-\! 1)\bar{q} \!-\! n_2\bar{q}^2}{2\bar{q}} = t_1(u_2) \!-\! 1.
    \label{eq:c2}
\end{equation}
}

Symmetrically, for each user $k\in V_3$ and  $k\in V_4$, since $\mathbf p$ is NE, we have
\begin{align}
t_2(u_1)-1 \leq u_2 \leq t_2(u_1),
\label{eq:c4}
\end{align}
where
\begin{align*}
   t_2(u_1) := \frac{q\mu/\phi + u_1(1+\bar{q}^2)+(n_2+1)\bar{q} - n_1\bar{q}^2}{2\bar{q}}.
\end{align*}
Note that Eq. (\ref{eq:c1_}) - (\ref{eq:c4}) are the sufficient and necessary conditions for an arbitrary strategy $\mathbf{p}$ to achieve NE. 
Now we are ready to give a characterization of NEs.  
\begin{theorem}\label{thm:cha}
Let $\mathbf p$ be an arbitrary strategy profile for the game with two sources. Let $u_1$ and $u_2$ be the number of users in $N_1$ and $N_2$ who choose DP under $\mathbf p$, respectively. We have
\begin{itemize}
    \item[{\rm [1]}] when (a) $u_1=n_1,u_2<n_2$, or (b) $u_1=0,u_2>0$, $\mathbf p$ cannot be a NE;
    \item[{\rm [2]}] when $u_1\in[0,n_1),u_2\in[0,n_2)$, $\mathbf p$ is NE if and only if $u_1\ge t_1(u_2)-1$ and $u_2\ge t_2(u_1)-1$;
    \item[{\rm [3]}] when $u_1\in(0,n_1),u_2=n_2$, $\mathbf p$ is NE if and only if $u_1\in[t_1(u_2)-1,t_1(u_2)]$;
    \item[{\rm [4]}] when $u_1=n_1,u_2=n_2$, $\mathbf p$ is NE if and only if $n_1\bar{q}\le q\mu/\phi+n_2+\bar{q}$.
\end{itemize}
\end{theorem}
\begin{proof}
Given $\mathbf p$, let $(V_1,V_2,V_3,V_4)$ be a partition of $N$ as defined above. We discuss the four cases. 

\textbf{Case 1.} When (a) $u_1=n_1$ and $u_2<n_2$,  $V_4$ is nonempty. If $\mathbf p$ is a NE, it must satisfy $t_2(u_1)-1\le u_2$.  
However, because $\bar{q}u_2\le n_1,\bar{q}u_2\le (n_2-1)\bar{q}$ and $q\mu/\phi>0$, it cannot hold. 

When (b) $u_1=0$ and $u_2>0$, $V_2$ is nonempty. If $\mathbf p$ is a NE, it must satisfy Eq. (\ref{eq:c2}), that is, $u_1\ge t_1(u_2)-1$. It follows that $0=2\bar{q}u_1\ge q\mu/\phi+u_2(1+\bar{q}^2)+(n_1-1)\bar{q}-n_2\bar{q}^2\ge q\mu/\phi+1+(n_1-1)\bar{q}-(n_2-1)\bar{q}^2\ge q\mu/\phi+1>0$, a contradiction. 

\smallskip\textbf{Case 2.} When $u_1\in[0,n_1),u_2\in[0,n_2)$, $V_2$ and $V_4$ are nonempty. It is easy to see that  $\mathbf p$ is NE if and only if $u_1\ge t_1(u_2)-1$ and $u_2\ge t_2(u_1)-1$ are satisfied simultaneously. 

\smallskip\textbf{Case 3.} When $u_2=n_2,u_1\in(0,n_1)$, $V_1,V_2,V_3$ are nonempty, and $V_4$ is empty.  $\mathbf p$ is NE if and only if $t_1(u_2)-1\le u_1\le t_1(u_2)$ and $u_2\le t_2(u_1)$ hold simultaneously. Moreover, note that $u_2\le t_2(u_1)$ is implied by $u_1\ge t_1(u_2)-1$. Therefore, the sufficient and necessary condition for NE is $u_1\in[t_1(u_2)-1,t_1(u_2)]$.

\smallskip\textbf{Case 4.} When $u_1=n_1,u_2=n_2$, $V_1,V_3$ are nonempty, and $V_2,V_4$ are empty.  $\mathbf p$ is NE if and only if $u_1\le t_1(u_2)$ and $u_2\le t_2(u_1)$. It is easy to see that, it is equivalent to $n_1\bar{q}\le q\mu/\phi+n_2+\bar{q}$.

\smallskip Note that every situation of $u_1,u_2$ is included in the above four cases. So we complete a characterization.
\end{proof}

Case 4 can be intuitively explained by considering the sidelink loss probability $q$ over link $(s_1,s_2)$. If $q$ is sufficiently high, no user would prefer the indirect path, and selecting the direct path  would be a NE for all users. Conversely, when there is no transmission loss over sidelink $(s_1,s_2)$ (i.e., $q=0$), every user would prefer to use the source with fewer users. Therefore, the profile of all users selecting DP is a NE only if the user distribution between the two sources is as even as possible, with $n_1\le n_2+1$. Based on Theorem \ref{thm:cha}, we give some interesting conclusions.
\begin{corollary}\label{cor:NE2}
If a strategy profile with $u_1=n_1,u_2=n_2$ is optimal, then it is also a NE. 
\end{corollary}
\begin{corollary}\label{cor:55}
A strategy profile with $u_1=0,u_2=0$ is a NE, if and only if (a) $n_1=n_2+1,\bar{q}=1$, or (b) $n_1=n_2,n_1(1-\bar{q}^2)\le \bar{q}-\frac{q\mu}{\phi}$.
\end{corollary}
Note that $u_1\ge t_1(u_2)-1$ and $u_2\ge t_2(u_1)-1$ cannot hold simultaneously when $q>\frac{2}{n}$, and $u_1\ge t_1(n_2)-1$ cannot hold when $n_1\bar{q}< q\mu/\phi+n_2+\bar{q}$.
\begin{corollary}\label{cor:uni}
When $n_1\bar{q}< q\mu/\phi+n_2+\bar{q}$ and $q>\frac{2}{n}$, the unique NE is that all users choose DP, i.e., $u_1=n_1,u_2=n_2$.
\end{corollary}
We end this section by  proving the existence of NE.
\begin{theorem}
For any game instance with two sources, there exists a NE with $u_1>0$ and $u_2=n_2$.
\end{theorem}
\begin{proof}
By Theorem \ref{thm:cha} (4), if $n_1\bar{q}\le q\mu/\phi+n_2+\bar{q}$, then the strategy profile that all users choose DP (i.e., $u_1=n_1,u_2=n_2$) is a NE. Otherwise, $n_1\bar{q}> q\mu/\phi+n_2+\bar{q}$. Let $\tilde m$ be an integer in interval $[\frac{q\mu/\phi+n_2+n_1\bar{q}-\bar{q}}{2\bar{q}},\frac{q\mu/\phi+n_2+n_1\bar{q}+\bar{q}}{2\bar{q}}]=[t_1(n_2)-1,t_1(n_2)]$, which always admits at least one integer. Note that $n_1>\frac{q\mu/\phi+n_2+n_1\bar{q}+\bar{q}}{2\bar{q}}\ge \tilde m>0$. By Theorem \ref{thm:cha}, a strategy profile with $u_1=\tilde m$ and $u_2=n_2$ is a NE. 
\end{proof}

\section{Numerical experiments}\label{sec:sim}

\begin{figure*}[htp]
     \centering
     \begin{subfigure}[b]{0.3\linewidth}
         \centering
         \includegraphics[width=2.1in]{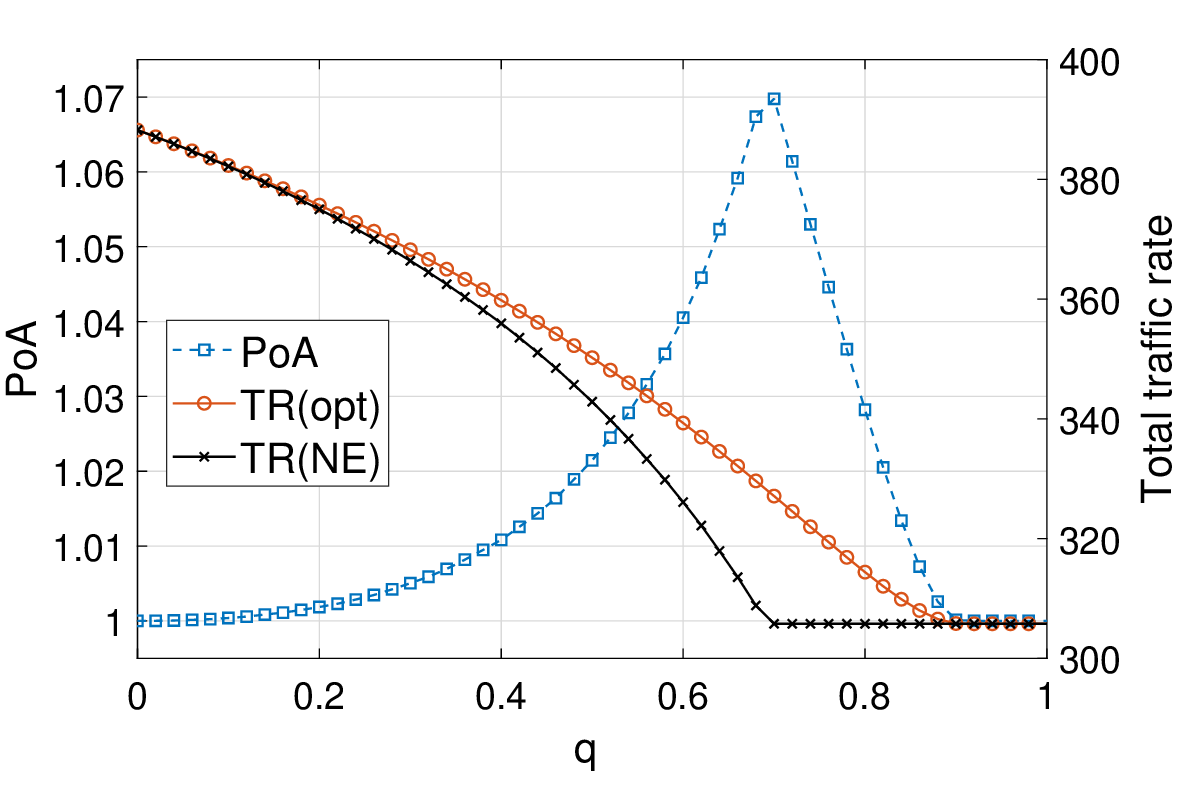}
        \caption{{$q\!\in\![0,1]$ with $n_1 \!= \!1000$, $ n_2\!=\!100$, $\phi\!=\!1$, $\mu=300$.}}
         \label{2p_q}
     \end{subfigure}
     \begin{subfigure}[b]{0.3\linewidth}
         \centering
         \includegraphics[width=2.1in]{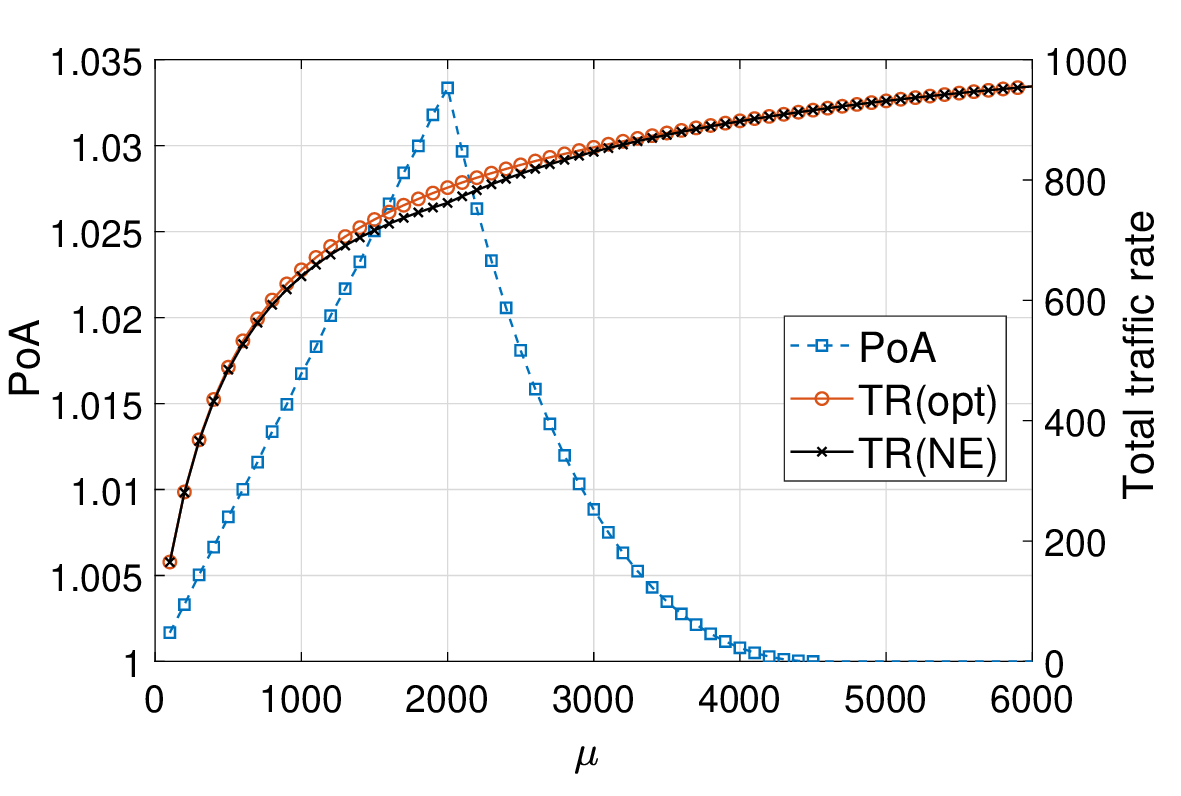}
         \caption{{$\mu\!\in\![1,6000]$ with $n_1\! =\! 1000$, $n_2\!=\!100$, $\phi\!=\!1$, $q\!=\!0.3$.}}
         \label{2p_mu}
     \end{subfigure}
     \begin{subfigure}[b]{0.3\linewidth}
         \centering
         \includegraphics[width=2.1in]{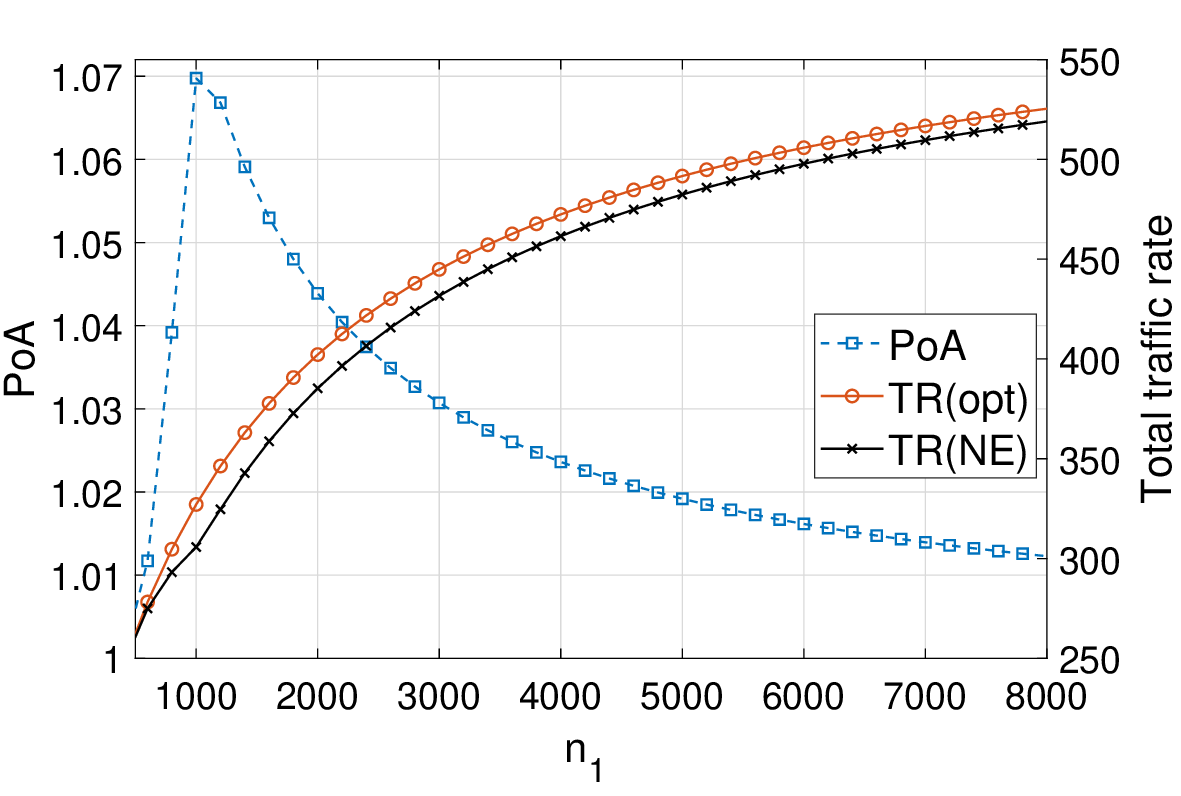}
         \caption{{$n_1\!\in\![500,8000]$, {$n_1$ is an integer} 
         with $n_2\!=\!100$, $\phi\!=\!1$, $\phi=1$, $q\!=\!0.7$.}}
         \label{2p_n1}
     \end{subfigure}
          \caption{ The PoA and total traffics in $opt$ and NE in two-source network 
         }
        \label{fig:2source}
\end{figure*}
Through numerical simulations, we explore the impact of traffic condition on network performance, i.e., the total traffic rate and PoA.
Recall that the traffic flow originating from each user is Poisson with rate $\phi$,
the service rate of each direct link is $\mu$, and the loss probability over each side  link is $q$. Assume $\phi=1$ for normalization.


 We first present the simulation results in two-source networks. In Fig.~\ref{fig:2source}, the PoA and the total traffics are plotted under different $q$, $\mu$ and $n_1$, showing a PoA of less than 1.08. Such a little gap between the optimal solution and the worst NE suggests that the gain of centralized-decision making over decentralized-decision making is trivial most of the time. As shown in Fig.~\ref{2p_q}, the total traffic decreases with the increase of $q$, i.e., the increased loss rate on sidelink. On the other hand, the $\text{PoA}$ is first increasing from $1$ at $q=0$, implying that the NE and optimal solution are the same with $u_1 = n_2 + v_1$. That is, we have the equal number of users on $(s_1,d)$ and $(s_2,d)$ in terms of both IP and DP. With the increase of $q$, the benefit of centralized-decision making is gradually unveiled. However, when $p$ reaches a certain value, the PoA goes down to $1$ quickly. An intuitive explanation is that, when $q$ becomes larger than the loss probability on DP, no users will choose IP in NE. And this strategy is optimal as well.


In Fig.~\ref{2p_mu}, the traffic rates grow with the increase of $\mu$ due to the increased probability of no congestion in \eqref{eq:loss_prob_direct}. This is, a high service rate help clear the collision and relieve congestion on both DP and IP. The PoA curve indicates that either in overloaded or less congested scenarios, there is little improvement of centralized-decision making. In Fig.~\ref{2p_n1}, the increased number of users leads to an increase of traffic rate in spite of the rise in loss rate. What is more, PoA tends towards $1$ for small and large $n_1$. As the strategies in $opt$ and NE are much similar for users at source $s_1$, i.e., DP in less biased scenario and IP severely biased scenario.
 


\begin{figure}[htp]
    \centering
    \includegraphics[scale=0.6]{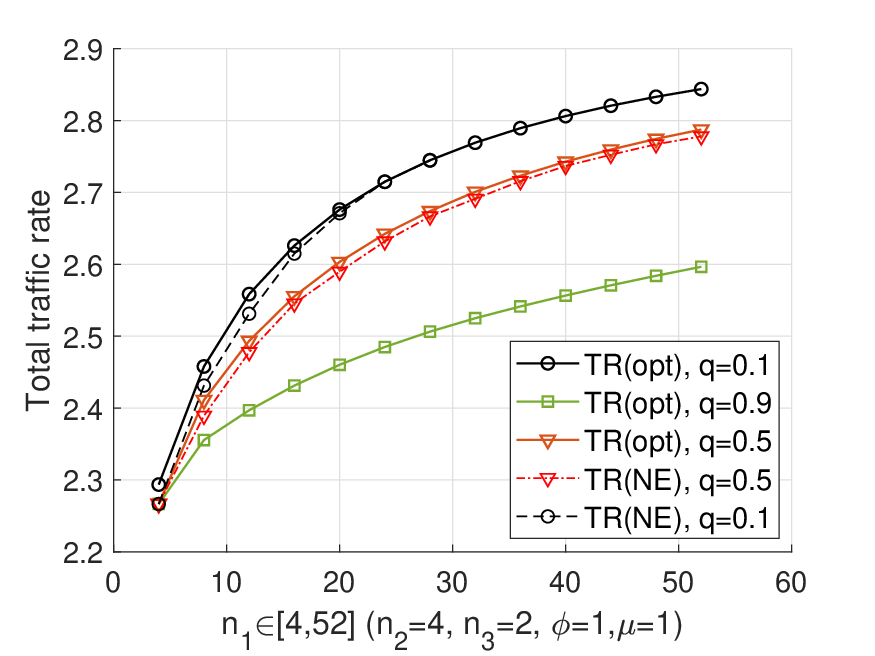}
    \caption{The impact of $n_1$ on total traffics in multi-source network. 
    }
    \label{multi_n1}
\end{figure}

\begin{figure}
    \centering
    \includegraphics[scale=0.6]{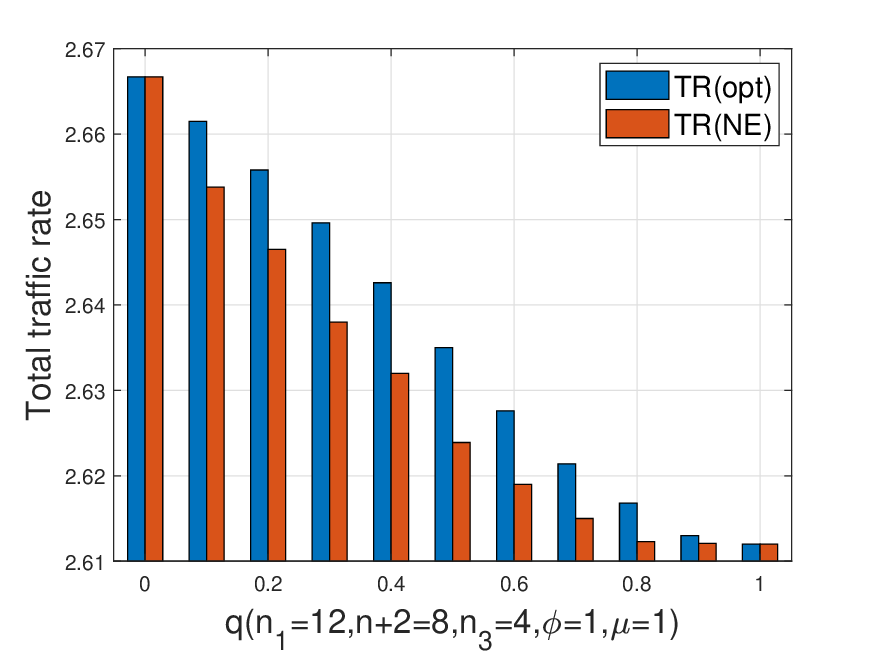}
    \caption{Comparison of NE strategy and optimal strategy in terms of traffic rate.}
    \label{3source_q_compare}
\end{figure}


In the multi-source network, while the optimal solution can be easily computed by Algorithm~\ref{alg:heuristic}, it is difficult to find all NEs even given Theorem~\ref{theo:msps}. {Hence,} we merely consider a small value of $m$ and $n$ (i.e., $m=3$). The service rate and traffic arrival rate are fixed as $\mu=1, \phi=1$. {Results are given in Fig.~\ref{multi_n1}, which shows similar results in Fig,~\ref{2p_n1}. It is obvious that the growth of the total traffic slows down gradually, because given the service rate, an increase of $n_1$ aggravates the network congestion. Second, the increase of loss rate on sidelink, leads to the increase of loss rate on IP. As a result, more users choose DP instead, which in turn worsens the network congestion. 



Figure \ref{3source_q_compare} plots the  performances  for a range of $q$. When $q=0$ and $q=1$, the PoA is exactly 1. 
The PoA converges to 1 when $q$ goes to 1, because when  the service rate is large enough compared with arrival rate, there is a sufficiently small congestion loss and all users like to choose DP.

\section{Conclusion}\label{sec:con}

In this work, we give a theoretical analysis of a load balancing game in cloud-enabled networks, in which the users want to minimize the loss probability of their packets with suitable routing strategies. In the centralized analysis, an efficient algorithm for maximizing the total traffic rate is proposed, according to Lemma~\ref{lem:both} and Lemma~\ref{lem:opt}. In the decentralized analysis, a characterization of  Nash equilibrium is given, and the PoA is investigated. Numerical experiments show that the efficiency loss due to selfish behaviors is relatively small in most cases.

There are many future directions that are worth exploring. First, we only focus on pure strategies of players in this work, and an immediate and natural question is how the users act when mixed strategies are allowed.  Second,  it would be interesting to investigate heterogeneous servers (source nodes) where each $s_i$ serves a different purpose or has a different service rate $\mu_i$. Moreover, while we only consider  direct path and one-hop indirect paths, a more general scenario where players can choose multi-hop indirect paths to the destination can be taken into consideration.



\bibliography{reference}
\bibliographystyle{IEEEtran}
\end{document}